\documentclass[12pt]{amsart}%
\usepackage{graphicx}
\usepackage{amsmath}
\usepackage{amsfonts}
\usepackage{amssymb}
\usepackage{fancyvrb}
\usepackage{longtable}
\newtheorem{theorem}{Theorem} 
\newtheorem{conjecture}[theorem]{Conjecture}
\newtheorem{lemma}[theorem]{Lemma}

\providecommand{\U}[1]{\protect\rule{.1in}{.1in}}
\usepackage{color}

\title{The Riccati Tontine: \\ How to Satisfy Regulators on Average}

\author{Moshe A. Milevsky and Thomas S. Salisbury}
\thanks{Milevsky is a Professor of Finance at the Schulich School of Business, York University, and Executive Director of the IFID Centre. Salisbury is a Professor in the Department of Mathematics and Statistics at York University. The authors acknowledge funding from the IFID Centre (Milevsky) and from NSERC (Salisbury) as well as helpful comments from seminar participants at the University of Amsterdam and the University of Leuven. Milevsky discloses that he is a consultant to Guardian Capital (Canada) and has a financial interest in matters discussed in this paper.}

\date{\today}

\begin{document}

\begin{abstract}


This paper presents a new type of modern accumulation-based tontine, called the Riccati tontine, named after two Italians: mathematician Jacobo Riccati (b. 1676, d. 1754) and financier Lorenzo di Tonti (b. 1602, d. 1684). The Riccati tontine is yet another way of pooling and sharing longevity risk, but is different from competing designs in two key ways. The first is that in the Riccati tontine, the representative investor is expected -- although not guaranteed -- to receive their money back if they die, or when the tontine lapses. The second is that the underlying funds within the tontine are deliberately {\em not} indexed to the stock market. Instead, the risky assets or underlying investments are selected so that return shocks are negatively correlated with stochastic mortality, which will maximize the expected payout to survivors. This means that during a pandemic, for example, the Riccati tontine fund's performance will be impaired relative to the market index, but will not be expected to lose money for participants. In addition to describing and explaining the rationale for this non-traditional asset allocation, the paper provides a mathematical proof that the recovery schedule that generates this financial outcome satisfies a first-order ODE that is quadratic in the unknown function, which (yes) is known as a Riccati equation.

\end{abstract}

\maketitle

\vspace{0.5in}

\begin{flushleft}
\end{flushleft}

\newpage

\section{Introduction and Motivation}
\label{sec:intro}

The academic research on pooled schemes that share longevity risks without providing any guarantees has seen significant growth in recent years. In particular, this issue in which this article appears is evidence of the ongoing interest in this field. This collection of papers, articles, and research strands that focus on pooling longevity risk has coalesced around the phrase: ``tontines'' although early literature in the 21st century used various actuarial terms such as ``participating life annuities'' or ``longevity-linked annuities'', or ``pooled annuities'', and ``self-annuitization schemes''. While this article does not aim to provide an extensive review of the growing field of tontine literature, the scholarly work on tontines can be classified into the following categories:

\begin{itemize}

\item Historical and/or economic analysis of 18th and 19th-century tontines, such as Jennings and Trout (1976), Ransom and Sutch (1987), Jennings, et al. (1988), Weir (1989), McKeever (2009), Rothschild (2009), Milevsky (2015), Gallais-Hamonno and Rietsch (2018), Hellwege (2018), as well as Li \& Rothschild (2020), and more recently McDiarmid (2024).

\item Simple one-period models of a basic tontine, such as the original work by Sabin (2010), Denuit and Vernic (2018), Fullmer and Sabin (2019), as well the recent work by Dhaene and Milevsky (2024), which provide basic intuition of the structure that is easily understood and accessible by non-actuaries.

\item Policy-focused articles that introduce and/or advocate for allowing (modern) tontines within 21st-century retirement plans, such as the work by Rotemberg (2009), Newfield (2014), Forman and Sabin (2015), Fullmer (2019), Winter and Planchet (2021), and especially Iwry et al. (2020).

\item Actuarial analysis of group self-annuitization schemes (GSA), longevity-indexed annuities (LIA), participating longevity-linked annuities (PLLA), all forms of a tontine, as well as articles on the impact of pool size on the stability of payouts. This would include the original work by Piggot et al. (2005), Stamos (2008), Qiao and Sherris (2013), Donnelly (2015), Bräutigam et al. (2017), and Pflaumer (2022).

\item Single cohort natural tontine schemes, introduced by Milevsky and Salisbury (2015), and subsequently expanded in Milevsky and Salisbury (2016).

\item Optimal lifecycle analysis of the allocation to (retirement) decumulation tontines and/or participating annuities, such as the work by: Maurer et al. (2013), Brautigam et al. (2017), Boon et al. (2020), and Chen, et al. (2020).

\item Tontine variants and generalizations, including elements of life insurance and long-term care that can satisfy bequest motives and/or hedge associated risks, such as the work by: Bernhardt and Donnelly (2019),  Chen et al. (2019, 2022), Dagpunar (2021), Hieber and Lucas (2022), Forsyth et al. (2024).

\end{itemize}

The unifying theme and underlying DNA of tontine schemes -- regardless of what they are called -- is that absolutely nothing is guaranteed. Likewise, these longevity risk-sharing arrangements live (or at least attempt to) outside the insurance regulatory environment and their stringent capital requirements. But alas, one can never escape regulation altogether and the recent promotors defaulted into the legal universe of securities. Moreover, the global trend has been for these regulators to impose a new type of constraint, namely that investors should receive their money back on average if forced to exit or sell early. In fact, anecdotal as well as behavioural evidence suggests that such a design – with a recovery, bequest or legacy value -- is more appealing to consumers as well. Accepting this condition dilutes a substantial portion of the mortality credits from pooling; however, rejecting it would mean the absence of regulatory approval.

\vspace{0.1in}

So, motivated by this engineering problem, our paper proves that under fairly general conditions, the optimal time-dependent recovery schedule -- which is one minus the surrender charge schedule -- imposed by such a requirement is unique (or nearly so) and satisfies a first-order ODE that is quadratic in the recovery function. We label this a Riccati Tontine and provide numerical values for the required recovery schedule using real-world parameters. Just to clarify, let's say the recovery schedule at time $t=10$ is $70\%$. In this case, if an investor dies at that time, they will only be able to pass on $70\%$ of the underlying account value to their heirs, and the remaining $30\%$ will be lost to surrender charges that are then shared as mortality credits among the remaining pool of investors. It's important to note that there is no guarantee that the $70\%$ of the underlying account value will be more than the original investment amount made by the participant at time $t=0$. This is a tontine after all, and designing based on averages is not making a guarantee. 

\vspace{0.1in}

In addition to the intrinsic mathematical elegance, our work guides how to design tontines that keep regulators happy (on average), but without wasting valuable mortality credits. The inspiration or intellectual motivation for this paper is the recent (re)introduction of tontines in Canada and, in particular, the GuardPath modern tontine product.\footnote{See: https://www.guardiancapital.com/investmentsolutions/guardpath-modern-tontine-trust/} That product includes a surrender charge schedule that {\em increases} over time, or a redemption schedule (percentage) that {\em declines} over time, similar to the curves we discuss and derive in this paper. Although to be absolutely crystal clear, there are many real-world features in the GuardPath product that differ from the idealized product we describe in this pages, and would should only view real-world incarnation of the Riccati tontine as motivation.

\subsection{Outline}

The remainder of this article is organized as follows. Section \ref{sectiontwo} introduces the basic stochastic model underlying the Riccati tontine, then Section \ref{sectionthree} introduces some variations in design, with both sections assuming a deterministic hazard rate. Section \ref{sectionfour} moves-on to introduce stochastic hazard rates and the proper selection of correlations and Section \ref{sectionfive} concludes the paper. 

\section{The Riccati tontine}
\label{sectiontwo}
\subsection{Preliminary model and the regulatory constraint}
\label{preliminarymodel}

We assume GBM dynamics for the underlying asset or mutual fund that represent the pure investment, which, as we will argue later in Section \ref{sectionfour}, need not be a broad-based index fund.
$$
dS_t=\mu S_t\,dt+\sigma S_t\,dB_t.
$$
Its dynamics are set and will not adjust to the mortality performance of the pool. One can think of the longevity risk sharing as an {\em overlay}, or simply a legal mechanism for sharing ownership units and rights of the underlying fund. The tontine manager or administrator will invest its funds fully in this single asset class which will then grow (if the drift is positive) and accumulate overtime. Recall that we are not describing decumulation scheme or one in which investors and funds can be added over time. Think of a closed pool in which $n$ investors allocate a fixed sum at time $t=0$, then wait for $T$ years hoping to survive, and then return at time $t=T$ to share the investment proceeds.

\vspace{0.1in}

Write $L_t$ for the total value of assets under management at time $t$, and $N_t$ for the number of survivors at time $t$. And without any loss of generality we assume that these $n$ investors all invest (only) \$1 initially, so $Z_t=\frac{L_t}{N_t}$ represents each of their individual account values, $L_0=N_0=n$ and $Z_0=1$. All of $N_t$, $L_t$, and $Z_t$ are right-continuous. We assume, for now, that members of the pool have a common deterministic and known hazard rate $\lambda_t$. We will modify this later, in Section \ref{sectionfour}, when we consider mortality-correlated assets.

\vspace{0.1in}

The pool or fund will remain in operation from time $t=0$ to a time horizon $T$ at which point the survivors will each be paid their share $Z_T$ of the fund. Early withdrawals (a.k.a. lapsation in the language of insurance) or deaths will be penalized by a surrender charge schedule.  In most cases, this will be $(1-k_t)$, where $0\le k_t\le 1$; in other words, someone lapsing or dying at $t<T$ is only paid $k_tZ_{t-}$. Of course, in the case of death, it is their estate or beneficiary that is paid this amount. The exception to the general rule is that we will find it convenient to vary this if only a single survivor remains. So the actual constraint will be that they are paid $K_tZ_{t-}$, where $K_t=k_t$ if $N_{t-}>1$, and $K_t=\kappa_t$ if $N_{t-}=1$. 

\vspace{0.1in}

The intuition here is that if (only) one person remains in the pool, it seems pointless to continue the tontine. Indeed, should that unlikely scenario occur the crowned winner of the tontine would be left alone in an empty risk pool, not accruing any more mortality credits, waiting for the terminal horizon to collect his or her winnings. In some sense if that happens the tontine element dies with the penultimate death. Arguably, the tontine fund should be dismantled as soon as the n'th minus one person dies or lapses. In fact, if the final (lucky) survivor happened to die before the terminal horizon, the fund would be left with a sum of money – more precisely, one minus $k$ times the pre-death fund value -- with no owners or investors to claim the assets. And, while, again, the probability of this scenario being realized is vanishingly small as the size of the pool goes to infinity, a properly specified longevity risk-sharing rule must account for all possibilities. Yes, bequeathing the leftovers to the fund managers -- or to an administrator, as suggested in a recent paper by Dhaene and Milevsky (2024) -- is certainly an option; it defeats the communal risk-sharing spirit of the tontine in addition to creating obvious moral hazard problems. Therefore, for the sake of both completeness and mathematical elegance, we will build and model our tontine fund to allow the last remaining survivor to receive the entire market value of the fund, dead or alive -- without incurring a recovery penalty. In some sense, this is a double win. They have survived everyone else in the original pool and they get to leave without incurring the financial bite of a recovery schedule.  

\vspace{0.1in}

The main purpose of the tontine is to provide a retirement nest egg to those who survive until a specific time or age. The surrender charge schedule ensures that mortality credits accrue to those who survive until this time, to the extent permitted by security regulators. As argued above, once the pool has only one survivor, no further mortality credits can accrue, and we will allow enhanced or full withdrawals at this point. In full generality, we do so via $\kappa_t$. We will require that  $0\le k_t\le\kappa_t\le 1$, with $\kappa_t=1$ being the case of full withdrawals by a lone survivor. Our Riccati tontine design focuses on finding appropriate choices of $k_t$ and $\kappa_t$.

\vspace{0.1in}

The presence of the surrender charge schedule means that it is sub-optimal to lapse, so we will, in fact, assume in our model that, unless we explicitly say otherwise, early recovery only occur because of death. In practice, there will still be a small background level of lapsation caused by liquidity demands, but we incorporate those into our numerical experiments by including a Makeham term in the hazard rate to account for lapsation over and above the Gompertz term that represents biological mortality (see Section \ref{Gompertzreview}). In any case, our model will, therefore, have at most one recovery at any time $t<T$. In other words, $\Delta N_t$ can only be $0$ or $-1$, and if $\Delta N_t=-1$ then:

$$
L_t=L_{t-}-K_tZ_{t-}=L_{t-}\Big(1-\frac{K_t}{N_{t-}}\Big)
$$
and 
$$
Z_t=\frac{L_{t-}-K_tZ_{t-}}{N_t}=Z_{t-}\Big(\frac{N_{t-}-K_t}{N_t}\Big)=Z_{t-}\Big(1+\frac{1-K_t}{N_t}\Big).
$$
In between jumps, of course, $Z_t$ is continuous and
$$
dZ_t=\mu Z_t\,dt+\sigma Z_t\,dB_t.
$$

The regulatory constraint is that individuals who die or lapse at any time must, on average, receive back their initial investment. There are several ways of making this precise, and we will adopt the following formulation (and will briefly discuss alternatives in Section \ref{sectionthree}). Split the pool into two imaginary groups, the first consists of one specific individual, who we call the representative agent, and the other being the rest of the pool. Let $\tau$ be the time the representative agent dies. Then the \emph{regulatory constraint} is taken to be the following:
\begin{equation}
E[K_tZ_{t-}\mid \tau=t]\ge 1\quad\text{for every $0\le t<T$}.
\label{basicconstraint}
\end{equation}

The objective of the tontine manager or administrator is to maximize -- or come close to maximizing -- the expected terminal payoff to survivors at time $T$, which can then be used to fund consumption over their remaining lifetime. In terms of our representative agent, this gives us an objective function $E[Z_T\mid\tau\ge T]$. We will explore this using the following notion:
A choice of $k_t$ and $\kappa_t$ satisfying the regulatory constraint is said to be \emph{extremal} if in fact
$$
\text{For every $0\le t<T$, either $k_t=0$ or }E[K_tZ_{t-}\mid \tau=t]= 1.
$$

\subsection{Conditioned model}
We have expressed our regulatory constraint and objective function by conditioning the original model. Now, we will restate these expressions using the conditioned processes as the models instead. This is the perspective mainly adopted in the rest of the paper. We used several conditionings, which resulted in several versions of the conditioned processes. However, we will switch between them as needed, relying on the context to clarify which version is being used.

\vspace{0.1in}

For the regulatory constraint at a given time $t$, we take the reference individual to have deterministic lifetime $t$, while the other $n-1$ individuals still die at rate $\lambda_s$, $0\le s\le t$. Then \eqref{basicconstraint} simply becomes that $E[K_tZ_{t-}]\ge 1$. It will be convenient to define $z_s=E[Z_{s-}]$ for $s\le t$, noting that these quantities are consistent, that is they do not actually depend on $t$. For the objective function, we take the representative agent to survive until time $T$, while the other $n-1$ individuals again die at rate $\lambda_t$, $0\le t\le T$. This makes our objective function simply $z_T$. We will also consider the limiting case $n=\infty$. Of course, in this case $\kappa_t$ is irrelevant, and $K_t=k_t$ for every $t$. Now $Z_t$ is continuous, and its dynamics are that 
\begin{equation}
dZ_t=\mu Z_t\,dt+\sigma Z_t\,dB_t+(1-k_t)\lambda_t Z_t\,dt
\label{limitingdynamics}
\end{equation}
(we refer interested readers to the Appendix for a derivation). 

\subsection{Main results}
Our primary focus will be the following tontines. Given a fixed hazard rate $\lambda_t$ for everyone in the tontine pool, we label a \emph{Riccati tontine} a tontine as above, with $k_t$ satisfying the Riccati equation
$$
k'_t+(\mu+\lambda_t)k_t-\lambda_tk_t^2,\quad k_0=1.
$$
See Figure \ref{FIG1} and Table \ref{TABLE1}. (All tables and figures are placed at the end of the paper.) Note that $\sigma$ does not appear in the above differential equation. Note also that for finite $n$ there can be multiple Riccati tontines, differing only in their choice for $\kappa_t$.
\begin{theorem}
\label{theorem1}
\begin{enumerate}
\item For $n=\infty$, the Riccati tontine satisfies the regulatory constraint and is extremal. 
\item For $n$ finite and $\kappa_t=1$, the Riccati tontine satisfies the regulatory constraint.
\item Moreover, for each $n$ finite, there is a choice of $\kappa_t\in [k_t,1]$ for which the Riccati tontine is extremal
\end{enumerate}
\end{theorem}
\begin{proof}[Proof of (a)] Solving \eqref{limitingdynamics} we have that 
\begin{equation}
\label{Zforinfiniten}
Z_t=e^{(\mu-\frac12 \sigma^2)t+\sigma B_t+\int_0^t (1-k_s)\lambda_s\,ds}
\end{equation}
so 
$$
E[k_tZ_t]=k_te^{\mu t+\int_0^t (1-k_s)\lambda_s\,ds}.
$$
We want this to $=1$, and differentiating, we see that the Riccati tontine fulfills this. 
\end{proof}
Note that when $n=\infty$, the (post-surrender charge) payout to a client dying at $t$ is therefore lognormal, namely $e^{\sigma B_t-\frac{\sigma^2}{2}t}$. We defer the proofs of (b) and (c) to the Appendix. Among the advantages of the Riccati tontine is that $k_t$ does not depend on $n$, so one does not need to know the number of purchasers before fixing the recovery schedule. This also makes $k_t$ exceptionally easy to compute -- other tontines that arise in this paper will typically require much longer compute time. Yet the Riccati tontine always satisfies the regulatory constraint and will turn out to be nearly optimal. The importance of extremal solutions is the following, which is proved in the Appendix.

\begin{theorem}
\label{maximality}
Assume that $k_t$ and $\kappa_t$ are continuous and yield an extremal tontine. Then $k_t$ maximizes the expected terminal payoff $z_T$ among tontines with this choice of $\kappa_t$. 
\end{theorem}

Our Riccati equation is also a 2nd order Bernoulli equation. A standard change of variables to $y_t=\frac{1}{k_t}$ leads to $y'_t=(\mu+\lambda_t)y_t-\lambda_t$, and then to the following explicit formula. 
\begin{equation}
\label{Bernoullisolution}
k_t=\frac{1}{y_t}=\frac{1}{1+\frac{\mu e^{\mu t}}{p_t}\int_0^t p_se^{-\mu s}\,ds},
\end{equation}
where $p_t$ is the survival probability $e^{-\int_0^t \lambda_s\,ds}$. More details are given in the Appendix, but note that this immediately implies that $0<k_t\le 1$ and that $k_t$ is decreasing to 0.

\subsection{Numerical results}
Refer to Table \ref{TABLE1} and consider the final expected payoff to survivors $z_T=E[Z_T]$. We will shortly see that this is well approximated by its limiting value $\frac{1}{k_T}$. For the Riccati tontine and the parameters considered in the table, this gives $z_T\approx 6.96238$, from an initial investment of 1. Compare this with $e^{\mu T}$, ie the appreciation factor without mortality, namely 4.0552; This shows the impact of mortality credits. 

\vspace{0.1in}

We show in the Appendix that the Riccati tontine is bracketed by the extremal tontines with $\kappa_t=1$ and with $\kappa_t=k_t$. These only differ when the pool is reduced to a single survivor, and with any reasonable size $n$ for the initial pool, the probability of this is extremely small. Therefore we expect that the choice of $\kappa_t$ has little impact. In other words, the Riccati tontine comes very close to optimizing $z_T$, regardless of $\kappa_t$. Table \ref{TABLE2} illustrates this, by showing $k_T$ and $z_T$ for the two bracketing tontines, as well as for the Riccati tontine, at various pool sizes. 

\vspace{0.1in}

Of course, the Riccati $k_t$ don't vary with pool size, but the $z_T$ do. Recall that $z_T$ is conditioned on the survival of the representative agent, which acts to increase the number of survivors and decrease the final payout, though this effect disappears once $n$ is large. Also, observe that $k_t$ can $=0$ with small pools -- in that case, the possibility of being the sole survivor and receiving a large payout is enough to satisfy the regulatory constraint, even if nothing is returned to someone who dies while others survive. See Figure \ref{FIGbracket} for a pictorial illustration of this. Finally, note that the Riccati $k_t$ is indeed bracketed by the others, that they are indistinguishable once $n=50$, and there is no practical difference even with $n=20$. Figure \ref{FIGbracket} shows the bracketing $k_t$'s with $n=3$ and $n=10$. In the latter case there are three curves, but the Riccati curve and the $\kappa_t=k_t$ extremal curve are too close to tell apart visually.

\vspace{0.1in}

For infinite $n$, \eqref{Zforinfiniten} immediately implies that the standard deviation of $Z_T$ is 
$z_T\sqrt{e^{\sigma^2T}-1}$.
In the appendix, we will calculate the standard deviations for finite pool sizes. By selecting a volatility of $\sigma=20\%$, we show these in Table \ref{TABLE3}. For the chosen parameters, the standard deviations are greater than the means and initially increase before decreasing as $n$ grows.

\section{Variations}
\label{sectionthree}

As stated earlier, the formulation of the regulatory constraint we have adopted is not the only possibility for making precise the general principle that early death or lapsation should return the initial investment ``on average''. Returning briefly to the preliminary model of Section \ref{preliminarymodel}, which is one in which $\lambda_t$ is still the hazard rate of the representative agent, one could examine conditions such as $E[k_\tau Z_\tau\mid \tau<T]\ge 1$, where $\tau$ is the lifetime of the representative agent. Or $E[k_{\tau\land t}Z_{\tau\land t}]\ge 1$ for every $t<T$, corresponding to regulating the expected payoffs from strategies of lapsing at time $t$, if one survives till then. In fact, one could note the condition as $E[k_tZ_t]\ge 1$ for every $t<T$, though this condition doesn't reference the death or lapsation of any specific individual. We did analyze the latter using the approach of Lemma \ref{odeforu} below, but will content ourselves by simply noting that the results are indistinguishable from the Riccati tontine, once the pool size is not tiny.

\vspace{0.1in}

Another possible variation in design that we will consider is one of not paying out deaths (or lapses) continuously, but rather according to some fixed schedule (eg once per year). All individuals dying in a given period would get paid at the same rate. Operationally, this should be much easier to administer. One question that could be considered is whether the estate of person A, who dies at the beginning of the period and has to wait the entire period in order to be paid out, will feel disadvantaged relative to the estate of person B, who dies at the end of the period and gets immediately paid out. We believe that there is likely little difference. It is true that if A had died slightly earlier, then the previous period's higher payout rate would have applied. On the other hand, dying earlier would mean foregoing both the investment returns for the current period, and the mortality credits from the previous period. These effects should roughly cancel out.

\vspace{0.1in}

A different question is whether the optimal continuous scheme differs significantly from the optimal discrete scheme. Figure \ref{FIG2} shows a plot with a tontine of size 20 and yearly payouts. The continuous (Riccati) values are a curve, and the discrete values are points. Even with this small pool size, the numbers are close: $k_1$ is 0.9315 (continuous) and 0.9324 (discrete); $k_{20}$ 0.1436 (continuous) and 0.1468 (discrete). So the final average payout to survivors ($\approx \frac{1}{k_T}$) is only slightly less in the discrete case. 

\section{Stochastic Mortality}
\label{sectionfour}

To this point in the paper, our model for the mortality hazard rate has been deterministic, which is the classical approach to thinking about $(p_s)$, and the Gompertz-Makeham model. But for over 20 years, the actuarial literature has gravitated towards a doubly stochastic model for mortality in which the mortality hazard rate is stochastic. See the work by Milevsky \& Promislow (2001) for one of the earliest papers in that stream, or Biffis (2005) where this is studied under the terms of an affine process for mortality. Within the context of tontines, if one allows for stochastic mortality that is correlated with financial markets, the question arises whether the underlying investment in the tontine fund should be selected in a manner that accounts for this correlation. 

\vspace{0.1in}

So, in this section, we'd like to investigate and understand the impact of this correlation on tontine returns, and in particular, to know which is preferable, a positive or negative correlation. A natural way to study this is to move to a model in which hazard rates follow processes similar to equity returns and to assume a correlation between the noises driving hazard rates and equity returns. The model we select is as follows:

\subsection{Model} 
Let $\Lambda_t$ be a Markov stochastic hazard rate of the form
$$
d\Lambda_t=\phi(t,\Lambda_t)\,dt + \psi(t,\Lambda_t)\,dW_t.
$$
The underlying asset that the tontine invests in will satisfy
$$
dS_t=\mu S_t\,dt + \sigma S_t\,dB_t
$$
and we assume a correlation $\rho$ between $W_t$ and $B_t$, ie $d\langle W,B\rangle_t=\rho\,dt$.

Because the hazard rate will often be exponentiated, we will generally assume that $\Lambda_t$ is bounded above.  See Huang et al. (2017) for further discussion of stochastic hazard rate models within the context of biological age and for references to the demographic literature discussing whether there is an upper bound for mortality rates. 

\vspace{0.1in}

We conjecture the following analogue of Theorem \ref{theorem1}. 

\begin{conjecture}
\label{RiccatiforSM}
Let $\Lambda_t$ be a general mortality model, correlated with stock returns as above, but with $\Lambda_t$ bounded above by $\lambda_\infty$. Let $k_t$ be the optimal recovery scheme in the case $n=\infty$. Then, the tontine using this $k_t$ and with $\kappa_t=1$ satisfies the regulatory constraint.
\end{conjecture}

\subsection{Asymptotics: cruise ships or funeral parlours?}
One can imagine two types of asset classes or investment portfolios in the context of stochastic mortality and correlations. The first could be described as investments (or companies) such as cruise ships, for example, that might benefit from unexpected growth in older people, that is, a decline in mortality rates. Cruise ships (stocks) would be negatively correlated with mortality. If there is an upwards shock to mortality rates (more deaths, fewer older people), there is less demand for cruise ships. On the other hand, funeral parlours would ``benefit'' from increases in mortality rates, and we might say that the correlations are positive with those types of assets. Now, we aren't advocating for investing in either, and it's an empirical question whether those correlations are as we hypothesize; yet it isn't a stretch to assume that investments with positive or negative correlations to mortality might exist.

\vspace{0.1in}

In our numerical experiments, we will assume a particular form for $\phi$ and $\psi$. If this was
$$
d\Lambda_t=(\Lambda_t-\eta)(g\,dt + \epsilon\,dW_t)
$$
then the solution $\Lambda_t=\eta + (\lambda_0-\eta)e^{(g-\frac12\epsilon^2)t+\epsilon W_t}$ would be a perturbation of Gompertz-Makeham mortality (we take $g=\frac{1}{b}+\frac{\epsilon^2}{2}$ and $\lambda_0=\eta+\frac{1}{b}e^{\frac{x-m}{b}}$). However, these dynamics will not yield a hazard rate that is bounded above. So we modify them by setting $\phi(t,\lambda)=0=\psi(t,\lambda)$ for $\lambda\ge\lambda_\infty$. In other words, we will follow GBM dynamics until the first time that $\Lambda_t$ exceeds some level $\lambda_\infty$, and then we will freeze the hazard rate, making 
$\Lambda_t=\lambda_\infty$ for subsequent $t$. For related models, see Huang et al. (2017) and Ashraf (2023).

It is hard to calculate $k_t$ with such models, without taking a simulation approach. Instead of doing that, we will adopt the specific mortality model described above, and as the notation suggests, treat $\epsilon$ as a small parameter. Taking $n=\infty$ this allows us to expand $k_t$ and $z_t$ in powers of $\epsilon$, to obtain asymptotics as $\epsilon\downarrow 0$. 

\vspace{0.1in}

We could have treated the variance of $Z_T$ similarly, to obtain an asymptotic Sharpe ratio. What we have chosen to do instead is to find the asymptotics of the utility of receiving $Z_T$. We can then analyze the dependence of this utility on $\rho$, and in particular, find whether positive $\rho$ or negative $\rho$ is preferable. 

Take $n=\infty$, so that
$$
Z_t=e^{(\mu-\frac12\sigma^2)t + \int_0^t(1-k_s)\Lambda_s\,ds + \sigma B_t}.
$$
Let $z_t=E[Z_t]$ and $k_t=\frac{1}{z_t}$, observing that $z_t$ is finite because we have assumed that $\Lambda_t$ is bounded. Using CRRA utility, with risk aversion $\gamma$, let $\theta$ be the utility of the final payout $Z_T$. In other words, $\theta=\frac{1}{1-\gamma}E[Z_T^{1-\gamma}]$. Expanding to first order, write $z_t=z^0_t(1+\epsilon\sigma\rho \bar z_t)+O(\epsilon^2)$ and $\theta=\theta_0(1+\epsilon\sigma\rho\bar \theta)$. We will shortly see that this normalization removes any dependence of  $\bar z$ and $\bar\theta$ on $\sigma$ and $\rho$. Then
$$
k_t=\frac{1}{z_t}=\frac{1-\epsilon\sigma\rho\bar z_t}{z^0_t}+O(\epsilon^2)=k^0_t(1-\epsilon\sigma\rho\bar z_t)+O(\epsilon^2).
$$
With the chosen mortality model, write $\lambda_t=\eta+(\lambda_0-\eta)e^{t/b}$ (recall that $g=\frac{1}{b}+\frac{\epsilon^2}{2}$) and assume that $\lambda_\infty>\lambda_T$. The expansions can be calculated via the following.
\begin{lemma}
\label{zbarexpansion}
$k^0_t$ is the Riccati tontine associated with $\lambda_t$, $z^0_t=\frac{1}{k^0_t}$, $\bar z_0=0$, and 
\begin{align*}
\bar z'_t&=\lambda_t[ t(1-k^0_t)+k^0_t\bar z_t]\\
\theta^0&=\frac{(z^0_T)^{1-\gamma}}{1-\gamma}e^{-\frac{\gamma(1-\gamma)}{2}\sigma^2T}\\
\bar\theta&=(1-\gamma)\int_0^T[(1-\gamma)(1-k^0_s)(\lambda_s-\eta)s+k^0_s\bar z_s\lambda_s]\,ds.
\end{align*}
\end{lemma}

Observe that $\bar z$ is positive, and that it grows in magnitude over time. So if $\rho>0$, this acts to increase the average final payoff, while having $\rho <0$ decreases it. For $0<\gamma<1$,  we have $\bar\theta >0$, so that positive $\rho$ increases utility, while negative $\rho$ decreases it. Risk-tolerant people (i.e. those with $0<\gamma<1$) will prefer a positive $\rho$. They will prefer a tontine that invests in funeral parlours over cruise ships (at least, when the tontine has a large pool, and when the randomness in hazard rates is small).

\vspace{0.1in}

When $\gamma>1$,  we have $\theta^0<0$, so utility rises when $\rho\bar\theta<0$ and it falls when $\rho\bar\theta>0$.  For $\gamma>1$ but close to 1, we can see that $\bar\theta>0$. But as we increase $\gamma$ this changes, and $\bar\theta$ turns negative. In other words, there is a $\gamma_0$ (independent of $\sigma$) such that cruise ships are preferred (conditional on survival) once $\gamma>\gamma_0$. For the parameters we have been working with, one can calculate numerically that $\gamma_0= 1.0598$;  Apparently, one doesn't have to be very risk averse before a negative $\rho$ is preferred. Even though the presence of negative correlation decreases average payoffs, the hedging effect it provides will increase utility. 

\vspace{0.1in}

For those who die early (at $t<T$), similar conclusions hold. Except that when $\gamma>\gamma_0$ and $t\approx 0$ the preference is for positive correlation, switching to a preference for negative correlation as $t$ rises. 

\vspace{0.1in}

For our numerical experiment, we choose $\sigma=20\%$. We don't have good data for $\epsilon$, but imagine that $\epsilon=15\%$;  then at age 75 ($t=10$ from initial age $x=65$) this means that one standard deviation $\pm$ changes the hazard rate by the same amount as ageing $\pm b\epsilon\sqrt{t}$, which with our parameters is $\pm 6.32$ years.  This does not seem unreasonable. 

\vspace{0.1in}

We don't have real-world calibrated data for $\rho$ either, but to speculate, suppose that 20\% of the market is closely linked to mortality. Suppose that a one-year mortality shock that ages the population by 4 years (so $\epsilon W_1b=4$) moves that segment of the market up by 30\% (and therefore the total market by the product, ie 6\%). Thus $e^{\sigma\rho W_1}=1.06$ and solving, we find that this corresponds to $\rho=10.9\%$; With those numbers, $\bar z_T=5.268$ and the expected final return moves from $z^0_T=6.962$ to $z_T\approx 7.083$, an increase of about $1.7\%$; So lower order in comparison to the tontine's basic mortality credits, but not negligible either.

\section{Conclusion}
\label{sectionfive}

This paper introduces and presents a new type of investment scheme called the Riccati tontine, which is a form of longevity risk pooling, named after the Italian mathematician Jacobo Riccati. Although some may question the need for yet another financial or insurance product, the Riccati tontine fills a gap in (or completes) the market by introducing a unique offering that cannot be easily replicated or synthesized by existing financial instruments or insurance policies. The Riccati tontine serves as an accumulation-based investment scheme and also functions as a decentralized longevity risk-sharing tool. Importantly, it differs from other designs for longevity risk pooling due to two novel engineering elements; the first is driven by securities regulation and the second by a personal hedging motive. The design of the Riccati tontine is such that the representative investor receives their money back on average, upon dying or lapsation. 

\vspace{0.1in}

Furthermore, in the Riccati tontine the underlying funds are invested in risky assets whose return shocks negatively correlate with mortality. This partially `hedges' the financial risk of the funds within the tontine by generating more death -- and thus higher expected payouts conditional on surivval -- in the economic scenarios where investment values decline. As a stylized example, one can consider an investment in companies that manage and arrange cruises for the elderly and retirees. Those cruise company stocks will likely decline in value {\em if} the demand for (old people) cruises falls, which is likely to occur during a pandemic. This means that the Riccati tontine fund is expected to perform worse during pandemics but yet will not lose participants' money on average. In sum, the central mathematical contribution of this paper is to prove that the recovery schedule that generates this dual outcome -- money back in expectation and a natural hedge for the fund when mortality is stochastic -- satisfies a first-order differential equation that is quadratic in the unknown function, a.k.a. a Riccati equation.

\section{Appendix}
\subsection{The limiting case $n=\infty$}
To understand the limiting dynamics \eqref{limitingdynamics}, consider large but finite $n$. Then approximately $\lambda_tN_t\Delta t$ deaths occur in a time interval $[t,t+\Delta t]$, each of which raises $Z_t$ by 
$$
\approx Z_t-Z_{t-}=Z_{t-}\frac{1-K_t}{N_t},
$$
for a net increment of $\approx Z_t(1-K_t)\lambda_t\Delta t$. 

\subsection{Hazard rates}
\label{Gompertzreview}
For those who are not familiar with the Gompertz-Makeham model, we recommend any actuarial textbook, such as Dickson, et al. (2019). In this paper, we use the parameterization of  Milevsky and Salisbury (2015) as it relates to tontines, whereby $\lambda_t=\eta +\frac{1}{b}e^{\frac{x+t-m}{b}}$. We've also used $p_t$ for a survival probability, so $p_t=e^{-\int_0^t \lambda_s\,ds}$. In the discrete section, $p_i$ is the survival probability for the $i$th interval, so $p_i=e^{-\eta\Delta-e^{\frac{x+(i-1)\Delta-m}{b}}(e^{\frac{\Delta}{b}}-1)}$

\subsection{Proof of Theorem \ref{theorem1}}
Fix $t<T$ as the lifetime of the representative agent, and let $s$ range over $[0,t)$. Set $u_j(s)=E[L_{s}1_{\{N_{s}=j\}}]$, and take the convention that $u_{n+1}(s)=0$. It is clear that the $u_j(s)$ remain consistent if we vary $t>s$.
\begin{lemma} For $1\le j\le n$ we have
$$
u'_j(s)=\mu u_j(s)+\lambda_s[j(1-\frac{k_s}{j+1})u_{j+1}(s)-(j-1)u_j(s)]
$$
\label{odeforu}
\end{lemma}
\begin{proof} To first order, 
\begin{align*}
&u_j(s+\Delta s)=E[L_{s+\Delta s},N_{s+\Delta s}=j]\\
&\qquad =E[L_{s+\Delta s},N_{s+\Delta s}=j, \Delta N_s=0]+E[L_{s+\Delta s},N_{s+\Delta s}=j, \Delta N_s=-1]\\
&\qquad =E[L_{s}+\mu L_s\Delta s+\sigma L_s\Delta B_s,N_s=j, \Delta N_s=0]\\
&\qquad\qquad\qquad+E[L_s+\mu L_s\Delta s+\sigma L_s\Delta B_s-\frac{k_sL_s}{j+1},N_s=j+1, \Delta N_s=-1]\\
&\qquad = u_j(s)(1+\mu\Delta s)(1-(j-1)\lambda_s\Delta s)+u_{j+1}(s)(1-\frac{k_s}{j+1}+\mu\Delta s)\lambda_sj\Delta s\\
&\qquad =u_j(s)+[u_j(s)(\mu-(j-1)\lambda_s)+u_{j+1}(s)\lambda_s (j+1-k_s)\frac{j}{j+1}]\Delta s.
\end{align*}
Now take limits.
\end{proof}
Note that the initial conditions are that $u_j(0)=n$ if $j=n$ and $=0$ otherwise.

\begin{proof}[Proof of Theorem \ref{theorem1}b]
Fix $t<T$ as the lifetime of the representative agent, and let $s$ range over $[0,t)$. If $N_\cdot$ jumps at $s$ and $N_{s-}\ge 2$ then $\Delta Z_s=\frac{1-k_s}{N_{s-}-1}Z_{s-}$. The jump intensity is $(N_{s-}-1)\lambda_s$, since there are $N_{s-}-1$ other individuals in the pool, if we exclude the representative agent. We have 
$$
dZ_s=\mu Z_s\,ds + \sigma Z_s\,dB_s+\Delta Z_s
$$
and the compensator of the jumps is
$$
\int_0^s (1-k_q)\lambda_qZ_q1_{\{N_q\ge 2\}}\,dq=\int_0^s (1-k_q)\lambda_q(Z_q-Z_q1_{\{N_q=1\}})\,dq.
$$
Therefore 
\begin{multline*}
Z_s=1+\int_0^s [\mu+(1-k_q)\lambda_q]Z_q\,dq+\int_0^s\sigma Z_q\,dB_q-\int_0^s(1-k_q)\lambda_qZ_q1_{\{N_q=1\}}\,dq + \\
+ \text{compensated sum of jumps}.
\end{multline*}
Taking expectations, we have that 
$$
z_s=1+\int_0^s [\mu+(1-k_q)\lambda_q]z_q\,dq-\int_0^s (1-k_q)\lambda_q u_1(q)\,dq,
$$
so
\begin{equation}
\label{zode}
z'_s=\mu z_s + (1-k_s)\lambda_s[z_s-u_1(s)].
\end{equation}
In general,
\begin{equation}
\label{kappaequation}
E[K_sZ_s]=k_s[z_s-u_1(s)] +\kappa_su_1(s)=k_sz_s+(\kappa_s-k_s)u_1(s),
\end{equation}
for $s<t$. Combining the above with Lemma \ref{odeforu} and the Riccati equation, we can compute $\frac{d}{ds}E[K_sZ_s]$. After some cancellation, we arrive at the following expression.
$$
\mu k_su_1(s)+[\kappa_s-k_s][\mu u_1(s)+\lambda_s(1-\frac{k_s}{2})u_2(s)] + \kappa'_su_1(s).
$$
Substituting $\kappa_s=1$, the last term drops out, while the other terms are $\ge 0$, so $E[K_sZ_s]$ is increasing. Taking limits as $s\uparrow t$, we see that $E[K_tZ_{t-}]\ge E[K_0Z_0]=1$, as required.
\end{proof}
Note that another approach to showing \eqref{zode} is to use the expression $z_s=\sum_{j=1}^n\frac{u_j(s)}{j}$ together with Lemma \ref{odeforu}.

\begin{proof}[Proof of Theorem \ref{theorem1}c]
The desired $\kappa$ is obtained by setting \eqref{kappaequation} equal to 1 and solving for $\kappa_t$ (note that $u_1$ is determined by $k$). This will automatically make $k_t$ extremal, provided we can verify that the imposed constraints hold, namely that $k_t\le\kappa_t\le 1$ for every $t$. Or in other words, that $k_t[z_t-u_1(t)]+k_tu_1(t)\le 1\le k_t[z_t-u_1(t)]+u_1(t)$, ie.
$$
k_tz_t\le 1\le k_tz_t+(1-k_t)u_1(t).
$$
To see the first inequality, differentiate $k_tz_t$ using \eqref{zode} and the Riccati equation. After some cancellations, we obtain that:
$$
(k_tz_t)'=-\lambda_tk_t(1-k_t)u_1(t)\le 0
$$
so $k_tz_t\le k_0z_0=1$. A similar calculation shows that 
$$
(k_tz_t+(1-k_t)u_1(t))'=\mu u_1(t)+\lambda_t(1-k_t)(1-\frac{k_t}{2})u_2(t)\ge 0
$$
so $k_tz_t+(1-k_t)u_1(t)\ge k_0z_0+(1-k_0)u_1(0)=1$, showing the second inequality.
\end{proof}

\subsection{Proof of Theorem \ref{maximality}} 

We now give the proof that extremal tontines maximize the expected final payoff.
\begin{proof}
Suppose $\ell_t$ is any recovery schedule, satisfying $\kappa_t\ge \ell_t\ge 0$ as well as the regulatory constraint. Let the $u_j$ satisfy Lemma \ref{odeforu} for $\ell_t$, and set $z_t=\sum_{j=1}^n\frac{u_j(t)}{j}$. The regulatory constraint can be expressed in the form
$\ell_t\sum_{j=2}^n \frac{u_j(t)}{j}+\kappa_t u_1(t)\ge 1$. Now define $\bar\ell_t$ by
\begin{equation}
\bar\ell_t=\max\Big(0,\frac{1-\kappa_t u_1(t)}{\sum_{j=2}^n \frac{u_j(t)}{j}}\Big), 
\label{extremaliteration}
\end{equation}
so $0\le\bar\ell_t\le \ell_t\le\kappa_t$. Then re-solve the equations of Lemma \ref{odeforu} using $\bar\ell_t$ instead of $\ell_t$, to give us new functions $\bar u_j(t)$ and $\bar z_t=\sum_{j=1}^n\frac{\bar u_j(t)}{j}$. By an ODE comparison theorem, $\bar u_j(t)\ge u_j(t)$ for each $j$, so the regulatory constraint still holds for the new tontine. Note that since $\kappa_t$ is continuous, so is $\bar\ell_t$. 

Now iterate this process (transfinitely if necessary), passing to the limit until $\bar z_T$ cannot be raised further (and still writing $\bar\ell_t$ etc. for the limiting objects). If there was an interval on which the regulatory constraint held with a strict inequality and $\bar \ell_t>0$, then we could repeat the process and strictly raise $\bar z_T$, which is a contradiction. Therefore on a dense set we have that either the constraint holds with equality or $\bar\ell_t=0$. The $\bar u_j$ are increasing limits of uniformly Lipschitz functions, by definition, so they are continuous. As a decreasing limit of continuous functions, $\bar\ell_t$ is upper semi-continuous, so in fact equality holds everywhere.  Therefore $\bar\ell_t$ is also continuous.  Since the original $\ell_t$ was arbitrary, this shows that a maximal solution exists and is extremal.

But for an extremal $\bar\ell_t$, \eqref{extremaliteration} in combination with Lemma \ref{odeforu} defines a system of Lipschitz ODE's (using here that the denominator is uniformly bounded below). Therefore there is a a unique solution starting from $\bar\ell_0=1$. Thus, there is only one extremal solution, namely the $k_t$ assumed in the theorem, which, therefore, is maximal.
\end{proof}

\subsection{Integral form for $k_t$}
If we let $u_t=e^{-\mu t}p_ty_t$  we get
$$
u'_t=-(\mu+\lambda_t)u_t +e^{-\mu t}p_ty'_t=-\lambda_te^{-\mu t}p_t.
$$
So $u_t-1=-\int_0^t \lambda_se^{-\mu s}p_s\,ds=e^{-\mu t}p_t-1+\mu\int_0^t p_se^{-\mu s}\,ds$, using integration by parts. This gives \eqref{Bernoullisolution}.

\subsection{Bracketing tontines}
The extremal tontine with $\kappa_t=k_t$ satisfies $k_tz_t=1$. Differentiating, we get
$$
k'_tz_t+\mu k_tz_t+\lambda_tk_t(1-k_t)(z_t-u_1(t))=0.
$$
Substituting $z_t=\frac{1}{k_t}$ and simplifying, we end up with 
$$
k'_t=-(\mu+\lambda_t) k_t+\lambda_tk_t^2+\lambda_tk_t^2(1-k_t)u_1(t).
$$
Since the last term is $\ge 0$ and ODE for the Riccati $k_t$ is that
$$
k'_t=-(\mu+\lambda_t) k_t+\lambda_tk_t^2,
$$
a standard ODE comparison theorem shows that this extremal $k_t$ is $\ge$ the Riccati one.

Likewise, the extremal tontine with $\kappa_t=1$ satisfies $k_t[z_t-u_1(t)]+u_1(t)=1$. Differentiating, we have
\begin{multline*}
k'_t[z_t-u_1(t)]+\mu k_t[z_t-u_1(t)]+\lambda_tk_t(1-k_t)[z_t-u_1(t)]\\
-\lambda_tk_t(1-\frac{k_t}{2})u_2(t)++\mu u_1(t)+\lambda_t(1-\frac{k_t}{2})u_2(t)=0.
\end{multline*}
Substituting $z_t-u_1(t)=\frac{1-u_1(t)}{k_t}$ and simplifying, we have
$$
k'_t=-(\mu+\lambda_t) k_t+\lambda_tk_t^2-\frac{k_t(1-k_t)}{1-u_1(t)}[\mu u_1(t)+\lambda_t(1-\frac{k_t}{2})u_2(t)].
$$
Since the last term is $\le 0$, the ODE comparison theorem implies that this extremal $k_t$ is $\le$ the Riccati one, as claimed.

Table \ref{TABLE2} involves computing $z_T$ for the extremal tontines with two choices of $\kappa_t$, namely $\kappa_t=k_t$ and $\kappa_t=1$, as well for the Riccati tontine. The latter is found by solving the equations of Lemma \ref{odeforu} numerically. To explore the former, note that being extremal means that
$$
1=k_t\sum_{j=2}^n\frac{u_j(t)}{j} + \kappa_t u_1(t)
$$
for $0\le t\le T$. Solving this for $k_t$ and substituting back into the equations of Lemma \ref{odeforu} gives a system of ODE's for the $u_j(t)$, which can be solved numerically.

\subsection{Variances}
Set $v_j(t)=E[L_t^2,N_t=j]$,  $j=1, \dots, n$. As in Lemma \ref{odeforu}, take $v_{n+1}=0$. Then to first order
\begin{align*}
&v_j(t+\Delta t)=E[L_{t+\Delta t}^2,N_{t+\Delta t}=j]\\
&\qquad =E[(L_t+\mu L_t\Delta t+\sigma L_t\Delta B_t)^2,N_t=j, \Delta N_t=0]\\
&\qquad\qquad\qquad+E[(L_t-\frac{k_tL_t}{j+1})^2,N_t=j+1, \Delta N_t=-1]\\
&\qquad = v_j(t)(1+(2\mu+\sigma^2)\Delta t)(1-(j-1)\lambda_t\Delta t)+v_{j+1}(t)(1-\frac{k_t}{j+1})^2\lambda_tj\Delta t\\
&\qquad =v_j(t)+[v_j(t)(2\mu+\sigma^2-(j-1)\lambda_t)+v_{j+1}(t)\lambda_tj(1-\frac{k_t}{j+1})^2]\Delta t.
\end{align*}
Or
$$
v'_j(t)=v_j(t)(2\mu+\sigma^2-(j-1)\lambda_t)+v_{j+1}(t)\lambda_tj(1-\frac{k_t}{j+1})^2.
$$
And now the variance of the representative agent's payout at time $T$ is
$$
E[\Big(\frac{L_T}{N_T}\Big)^2]-E\Big[\frac{L_T}{N_T}\Big]^2=\sum_{j=1}^n\frac{1}{j^2}v_j(T)-z_T^2.
$$

\subsection{Periodic payouts}
Instead of deaths triggering immediate payouts, suppose that deaths trigger payouts over $M$ periods. Then $\Delta=\frac{T}{M}$ is the length of each period. All deaths in a period get paid out simultaneously at the end of the period, so the order of death within the period does not matter. Thus the payout function consists of a sequence of values $k_1, \dots, k_M$. In the final year, the deaths get paid out at the rate of $k_M$, and then the survivors split the remainder. To keep things simple, we will not adjust the payout even when deaths exhaust the pool and leave some money on the table.

We have the following dynamics: 
$$
L_0=n, \quad L_{i}=L_{i-1}(e^{(\mu-\frac12\sigma^2)\Delta+\sigma \sqrt{\Delta}Z_i}-k_{i}\frac{N_{i-1}-N_{i}}{N_{i-1}})
$$
where the $Z_i$ are standard normal. Here $L_i$ and $N_i$ denote the total assets and survivors at the end of the $i$th period (ie after deaths are counted and payouts have been subtracted). 

We will carry out calculations using $u_j(i)=E[L_i,N_i=j]$, where in the expectation, we require the representative agent to survive the $i$th period. We take the convention that $u_j(0)=0$ if $j<n$ and $n$ otherwise.
Therefore
\begin{align}
u_j(i)&=\sum_{\ell=j}^n \binom{\ell-1}{j-1}p_i^{j-1}(1-p_i)^{\ell-j}u_\ell(i-1)[e^{\Delta \mu}-k_{i}(\frac{\ell-j}{\ell})] \nonumber \\
&=\frac1{p_i}\sum_{q=0}^{n-j}\text{NB}(q,j,p_i)u_{j+q}(i-1)[e^{\Delta \mu}-k_{i}(\frac{q}{j+q})]
\label{periodicdynamics}
\end{align}
for $1\le i\le M$, where $p_i$ is the period survival probability, and the NB are probabilities for the negative binomial distribution. We are also using that if the agent survives through period $i$, then they naturally survived period $i-1$ as well. 

For the constraint on $k_i$, we must assume the representative agent dies in the $i$th period. Then just before the $i$th payout, the total assets are $\tilde L_i=L_{i-1}e^{(\mu-\frac12\sigma^2)\Delta+\sigma \sqrt{\Delta}Z_i}$, so the payout received by the agent is $k_i\frac{\tilde L_i}{N_{i-1}}$. Our constraint is that this has mean $=1$. In other words, that 
\begin{equation}
1=k_{i} e^{\Delta\mu}\sum_{j=1}^n\frac{u_j(i-1)}{j}
\label{periodicconstraint}
\end{equation} 
for $1\le i\le M$. Going back and forth between \eqref{periodicdynamics} and \eqref{periodicconstraint} lets us find the $k$'s as in Figure \ref{FIG2}. In particular, $k_1=e^{-\Delta \mu}$.

\subsection{Correlated mortality: asymptotics}
\begin{proof}[Proof of Lemma \ref{zbarexpansion}]
Because $\Lambda_t$ is bounded, $\lambda_\infty>\lambda_T$, the probability of $\Lambda_t$ reaching this bound is $O(\epsilon^2)$. Therefore 
$$
Z_t=e^{(\mu-\frac12\sigma^2)t+\int_0^t(1-k_s)[\lambda_s+\epsilon (\lambda_s-\eta)W_s])\,ds + \sigma B_t}+O(\epsilon^2).
$$
except on an event of probability $O(\epsilon^2)$. The exponent is normally distributed with mean $(\mu-\frac12\sigma^2)t+\int_0^t(1-k_s)\lambda_s\,ds$, and variance
$$
E[\sigma^2B_t^2+2\sigma\epsilon\int_0^t(\lambda_s-\eta)(1-k_s)B_tW_s\,ds]+O(\epsilon^2)
=\sigma^2t+2\sigma\epsilon\int_0^t(\lambda_s-\eta)(1-k_s)\rho s\,ds+O(\epsilon^2).
$$
Therefore 
\begin{align*}
z_t&=e^{(\mu-\frac12\sigma^2)t+\int_0^t(1-k_s)\lambda_s\,ds}e^{\frac{\sigma^2t}{2}+\epsilon\sigma\rho\int_0^t s(\lambda_s-\eta)(1-k_s) \,ds}+O(\epsilon^2)\\
&=e^{\mu t+\int_0^t(1-k_s)[\lambda_s+\epsilon\sigma\rho s(\lambda_s-\eta)]\,ds}+O(\epsilon^2)\\
&=e^{\mu t+\int_0^t(1-k^0_s)\lambda_s\,ds} e^{\epsilon \int_0^t[(\lambda_s-\eta)\sigma\rho s(1-k^0_s)+\lambda_sk^0_s\sigma\rho\bar z_s]\,ds}+O(\epsilon^2)\\
&=z^0_t(1+\epsilon\sigma\rho \int_0^t[(\lambda_s-\eta) s(1-k^0_s)+\lambda_sk^0_s\bar z_s]\,ds)+O(\epsilon^2).
\end{align*}
Therefore
$$
\bar z_t=\int_0^t[(\lambda_s-\eta) s(1-k^0_s)+\lambda_sk^0_s\bar z_s]\,ds\quad \text{or}\quad \bar z'_t=(\lambda_t-\eta) t(1-k^0_t)+\lambda_tk^0_t\bar z_t
$$
as claimed. 

Turning to utility, we have that except on an event of probability $O(\epsilon^2)$,
\begin{align*}
Z_t^{1-\gamma}&=e^{(1-\gamma)[(\mu-\frac12\sigma^2)t+\int_0^t(1-k_s)[\lambda_s+\epsilon (\lambda_s-\eta)W_s])\,ds + \sigma B_t]}+O(\epsilon^2)\\
&=e^{(1-\gamma)[(\mu-\frac12\sigma^2)t+\int_0^t(1-k_s)\lambda_s)\,ds]}\cdot e^{\epsilon (1-\gamma)\int_0^t(1-k_s)(\lambda_s-\eta)W_s\,ds + \sigma(1-\gamma) B_t}+O(\epsilon^2)
\end{align*}
so
\begin{align*}
E[Z_t^{1-\gamma}]&=e^{(1-\gamma)[(\mu-\frac12\sigma^2)t+\int_0^t(1-k_s)\lambda_s\,ds]}\cdot E[e^{\epsilon (1-\gamma)\int_0^t(1-k_s)(\lambda_s-\eta)W_s\,ds + \sigma(1-\gamma) B_t}]+O(\epsilon^2)\\
&=e^{(1-\gamma)[(\mu-\frac12\sigma^2)t+\int_0^t(1-k_s)\lambda_s\,ds]}\cdot e^{\epsilon \sigma\rho(1-\gamma)^2\int_0^t(1-k_s)(\lambda_s-\eta)s\,ds + \frac12\sigma^2(1-\gamma)^2 t}+O(\epsilon^2)\\
&=e^{-\frac{\gamma(1-\gamma)}{2}\sigma^2t}\cdot e^{(1-\gamma)[\mu t+\int_0^t(1-k_s)\lambda_s\,ds]}\cdot e^{\epsilon \sigma\rho(1-\gamma)^2\int_0^t(1-k_s)(\lambda_s-\eta)s\,ds}+O(\epsilon^2)\\
&=e^{-\frac{\gamma(1-\gamma)}{2}\sigma^2t}\cdot e^{(1-\gamma)[\mu t+\int_0^t(1-k^0_s)\lambda_s\,ds]}\cdot e^{\epsilon\sigma\rho \int_0^t[(1-\gamma)^2(1-k^0_s)(\lambda_s-\eta)s+(1-\gamma)k^0_s\bar z_s\lambda_s]\,ds}+O(\epsilon^2)
\end{align*}
Setting $t=T$, we get
\begin{align*}
\theta^0&=\frac{1}{1-\gamma}e^{-\frac{\gamma(1-\gamma)}{2}\sigma^2T}\cdot e^{(1-\gamma)[\mu T+\int_0^T(1-k^0_s)\lambda_s\,ds]}
=\frac{(z^0_T)^{1-\gamma}}{1-\gamma}e^{-\frac{\gamma(1-\gamma)}{2}\sigma^2T}\\
\bar\theta&=(1-\gamma)\int_0^T[(1-\gamma)(1-k^0_s)(\lambda_s-\eta)s+k^0_s\bar z_s\lambda_s]\,ds
\end{align*}
as required.
\end{proof}
For completeness, we'll also work out the case of logarithmic utility (ie $\gamma=1$). 
\begin{align*}
E[\log Z_T]&=(\mu-\frac12\sigma^2)T+\int_0^T(1-k_s)\lambda_s\,ds +O(\epsilon^2)\\
&=(\mu-\frac12\sigma^2)T+\int_0^T(1-k^0_s)\lambda_s\,ds+\epsilon\rho\sigma\int_0^Tk^0_s\bar z_s\lambda_s\,ds+O(\epsilon^2).
\end{align*}
The analogue of $\bar\theta$ is therefore $\int_0^Tk^0_s\bar z_s\lambda_s\,ds>0$. 
So with logarithmic utility, $\rho>0$ is always favourable (at least asymptotically), as would be expected from our earlier calculations.

\clearpage

\begin{table}
\begin{center}
\begin{tabular}{|c|c|}
\hline
year $t$ & $k_t$ \\
\hline
1	&	0.93147	\\
2	&	0.86589	\\
3	&	0.80327	\\
4	&	0.74360	\\
5	&	0.68686	\\
6	&	0.63300	\\
7	&	0.58198	\\
8	&	0.53372	\\
9	&	0.48819	\\
10	&	0.44527	\\
11	&	0.40492	\\
12	&	0.36704	\\
13	&	0.33155	\\
14	&	0.29838	\\
15	&	0.26744	\\
16	&	0.23866	\\
17	&	0.21196		\\
18	&	0.18727	\\
19	&	0.16451	\\
20	&	0.14363	\\
\hline
\end{tabular}
\end{center}
\bigskip
\caption{recovery schedule $k_t$ in a $T=20$ year Riccati tontine for a pool of $x=65$-year-olds. We assume Gompertz $m=90$ and $b=10$, with a Makeham parameter $\eta=2\%$, under an expected investment return of $\mu=7\%$.}
\label{TABLE1}
\end{table}


\clearpage

\begin{table}
\begin{center}
\begin{tabular}{|c|cc|cc|cc|}
\hline
pool size & \multicolumn{2}{c|}{$\kappa_t=1$, $k_t$ extremal} & \multicolumn{2}{c|}{Riccati $k_t$} & \multicolumn{2}{c|}{$\kappa_t=k_t$, $k_t$ extremal} \\
 $n$ & $k_{20}$ & $z_{20}$ & $k_{20}$ & $z_{20}$ & $k_{20}$ & $z_T{20}$ \\
\hline
2 & 0 & 5.78882 & 0.143629 & 5.33605 & 0.188823 & 5.29598 \\
3 & 0 & 6.48671 & 0.143629 & 6.02782 & 0.166672 & 5.99979 \\
5 & 0 & 6.92345 & 0.143629 & 6.64347 & 0.150730 & 6.63437 \\
10 & 0.117374 & 6.96237 & 0.143629 & 6.93912 & 0.144120 & 6.93868 \\
20 & 0.143352 & 6.96237 & 0.143629 & 6.96224	 & 0.143632 & 6.96224 \\
50 & 0.143629 & 6.96237 & 0.143629 & 6.96237 & 0.143629 & 6.96237\\
$\infty$ & 0.143629 & 6.96238 & 0.143629 & 6.96238 & 0.143629 & 6.96238\\
\hline
\end{tabular}
\end{center}
\bigskip
\caption{Terminal recovery values $k_T$ and average appreciation factor $z_T$ in $T=20$ year tontines for various pools of $x=65$-year-olds. We assume Gompertz $m=90$ and $b=10$, with a Makeham parameter $\eta=2\%$, under an expected investment return of $\mu=7\%$. On the left we show the extremal tontine with $\kappa_t=1$. The Riccati tontine is in the middle. On the right we show the extremal tontine with $\kappa_t=k_t$.}
\label{TABLE2}
\end{table}

\clearpage

\begin{table}
\begin{center}
\begin{tabular}{|c|c|}
\hline
pool size $n$ & std. dev. \\
\hline
2	&	6.215	\\
3      &      7.209        \\
5	&	8.123	\\
10	&	8.332	\\
20	&	8.004	\\
50	&	7.812	\\
100	&	7.758	\\
200	&	7.732	\\
500	&	7.717	\\
1000 &     7.713        \\
$\infty$ &	7.708	\\
\hline
\end{tabular}
\end{center}
\bigskip
\caption{Standard deviation of the payout for a representative investor in a $T=20$ year horizon Riccati tontine for a pool of size $n$ with $x=65$-year-olds. We assume Gompertz $m=90$ and $b=10$, with a Makeham parameter $\eta=2\%$, under an expected investment return of $\mu=7\%$ and volatility $\sigma=20\%$.}
\label{TABLE3}
\end{table}

\clearpage

\begin{figure}
\begin{center}
\includegraphics[width=0.45\textwidth]{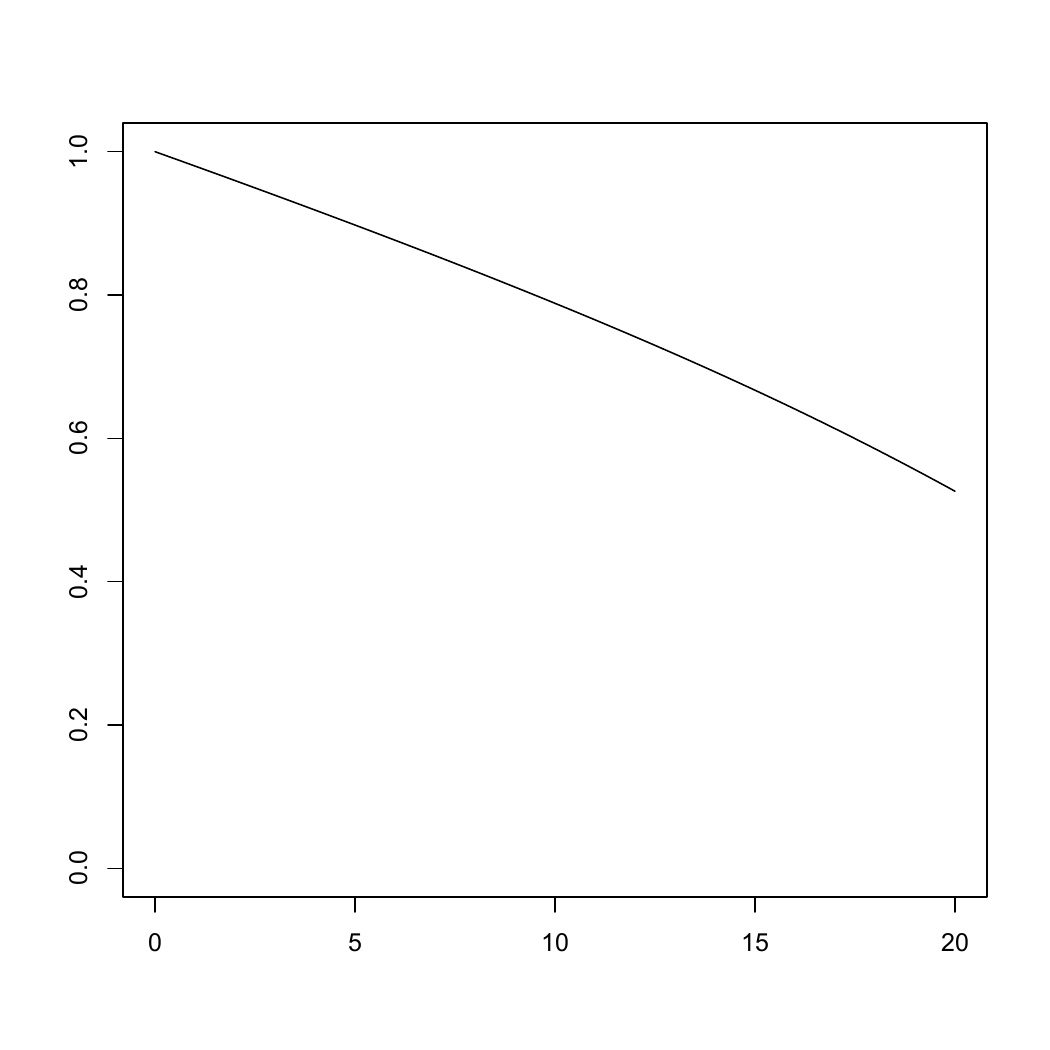} 
\includegraphics[width=0.45\textwidth]{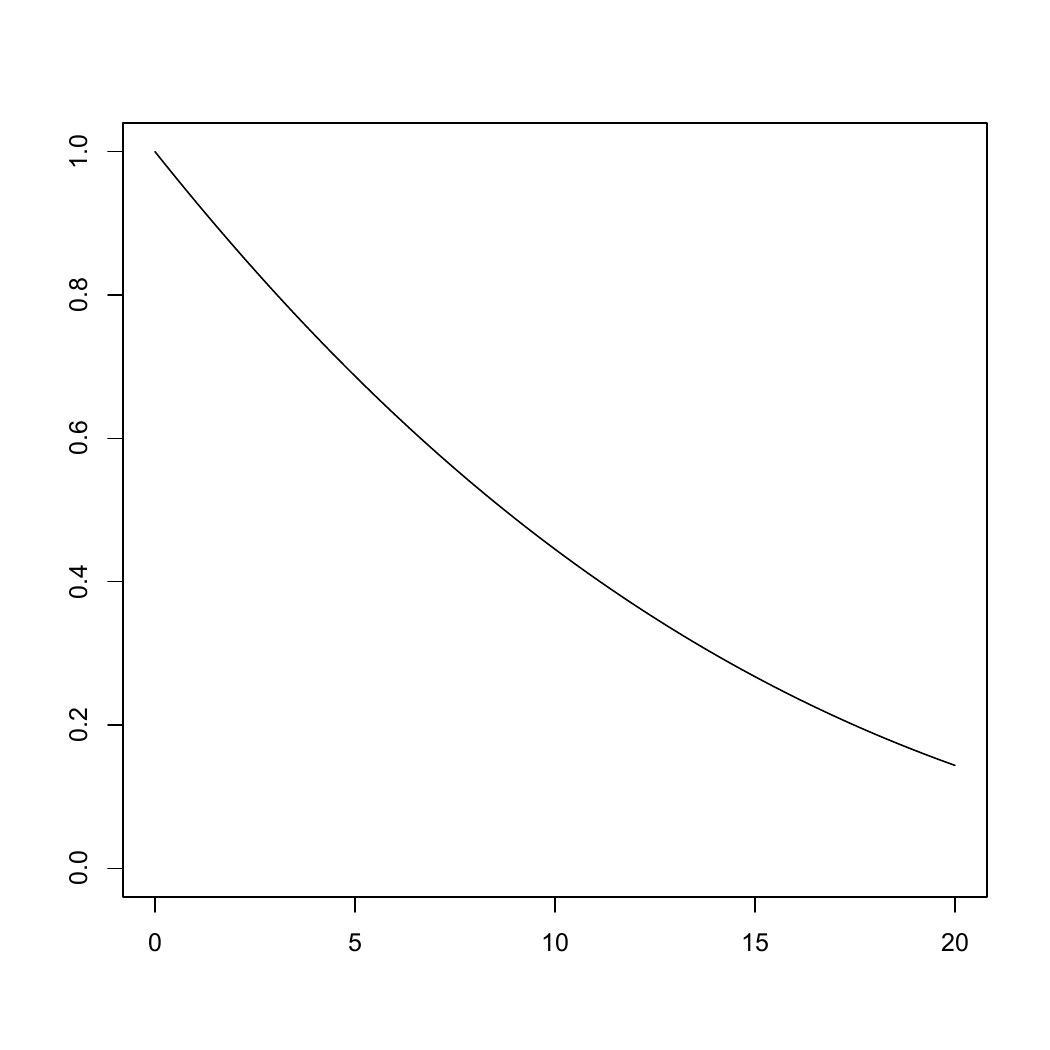} 
\caption{Plot of the optimal recovery schedule $k_t$ in a $T=20$ year Riccati tontine for a pool of $x=65$-year-olds, assuming Gompertz $m=90$ and $b=10$, with a Makeham parameter $\eta=2\%$, under an investment return of $\mu=2\%$ (left, suitable for fixed income investments) and $\mu=7\%$ (right, suitable for a balanced portfolio).}
\label{FIG1}
\end{center}
\end{figure}

\clearpage

\begin{figure}
\begin{center}
\includegraphics[width=0.45\textwidth]{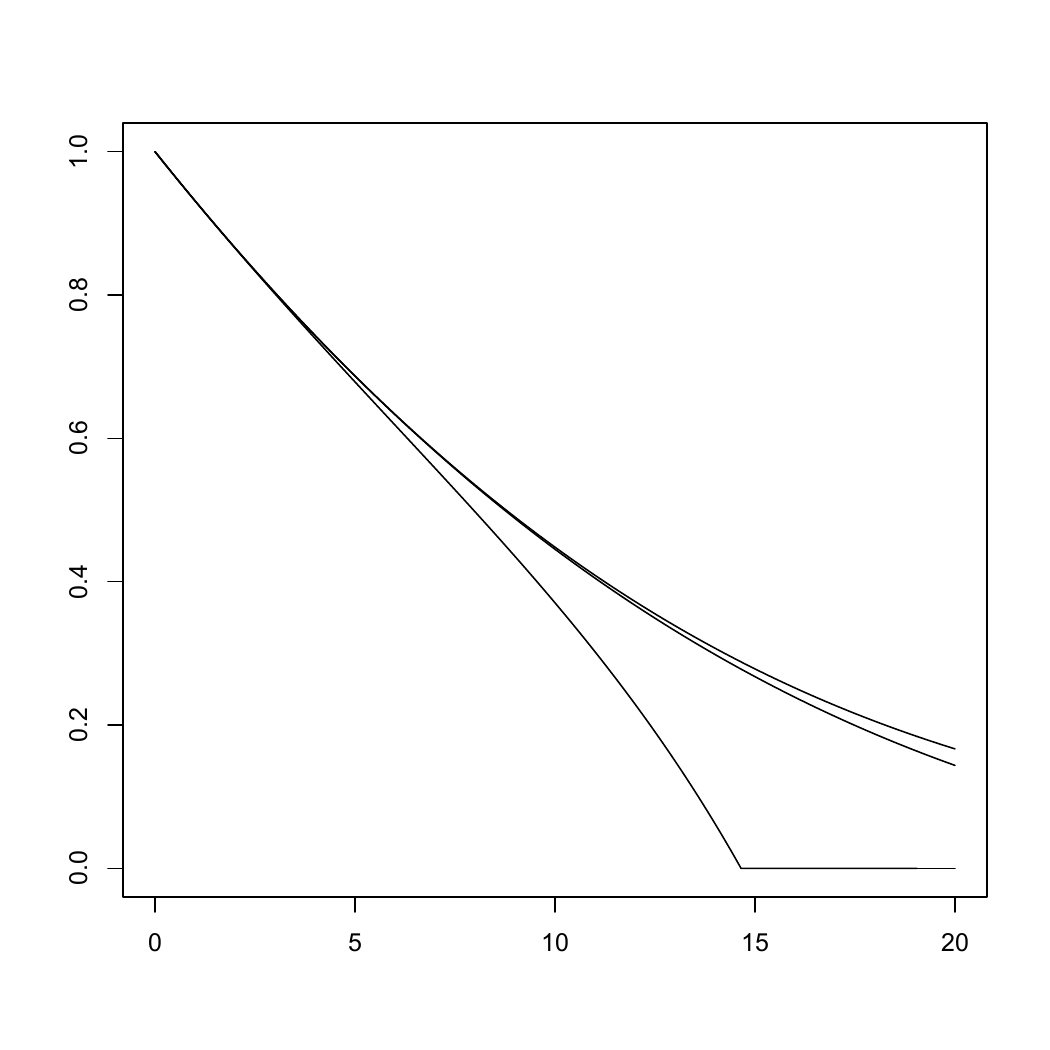} 
\includegraphics[width=0.45\textwidth]{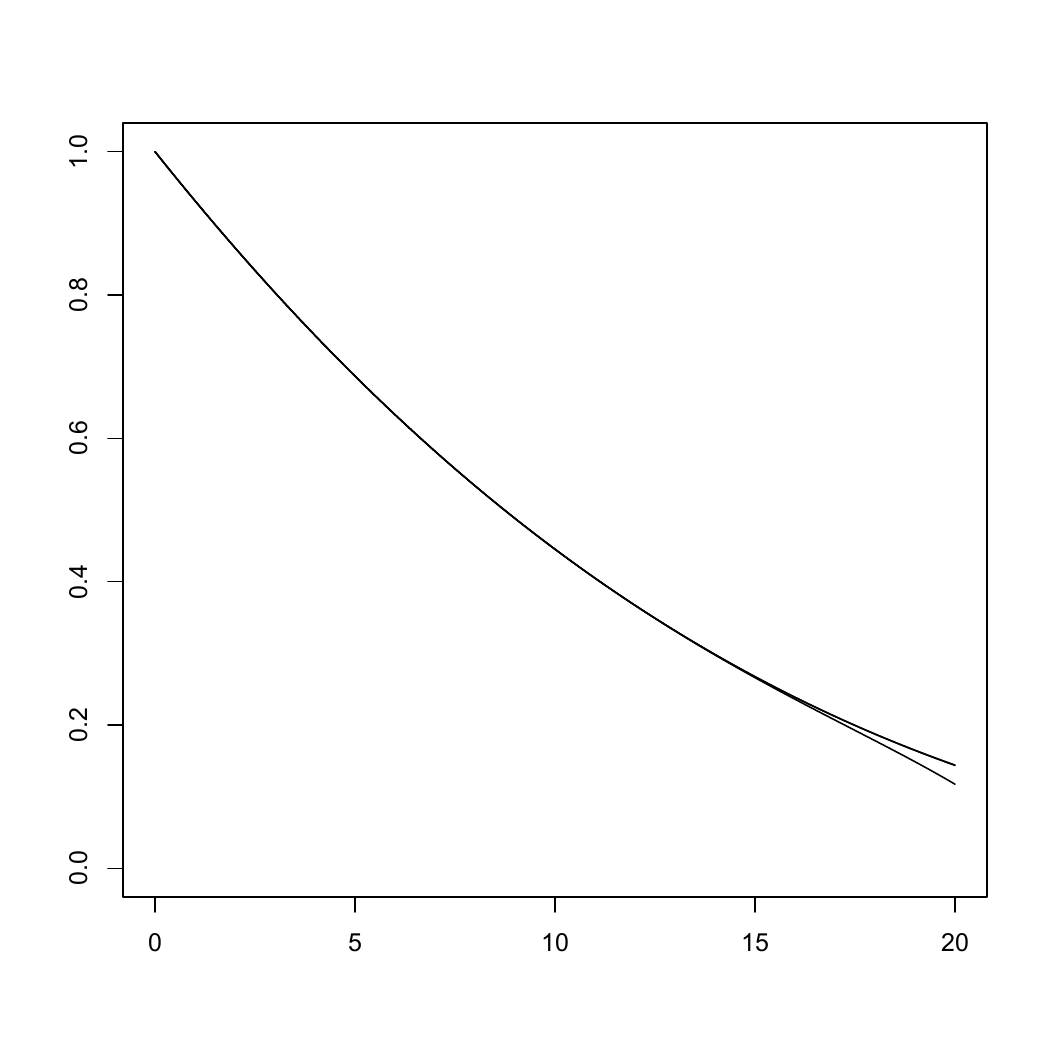} 
\caption{Plot of recovery schedules $k_t$ in $T=20$ year tontines for a pool of $x=65$-year-olds, assuming Gompertz $m=90$ and $b=10$, with a Makeham parameter $\eta=2\%$, under an investment return of $\mu=7\%$. Showing the Riccati tontine and extremal bracketing tontines on either side of it. On the left, the pool size $n=3$; on the right, the pool size $n=10$.}
\label{FIGbracket}
\end{center}
\end{figure}

\clearpage

\begin{figure}
\begin{center}
\includegraphics[width=0.60\textwidth]{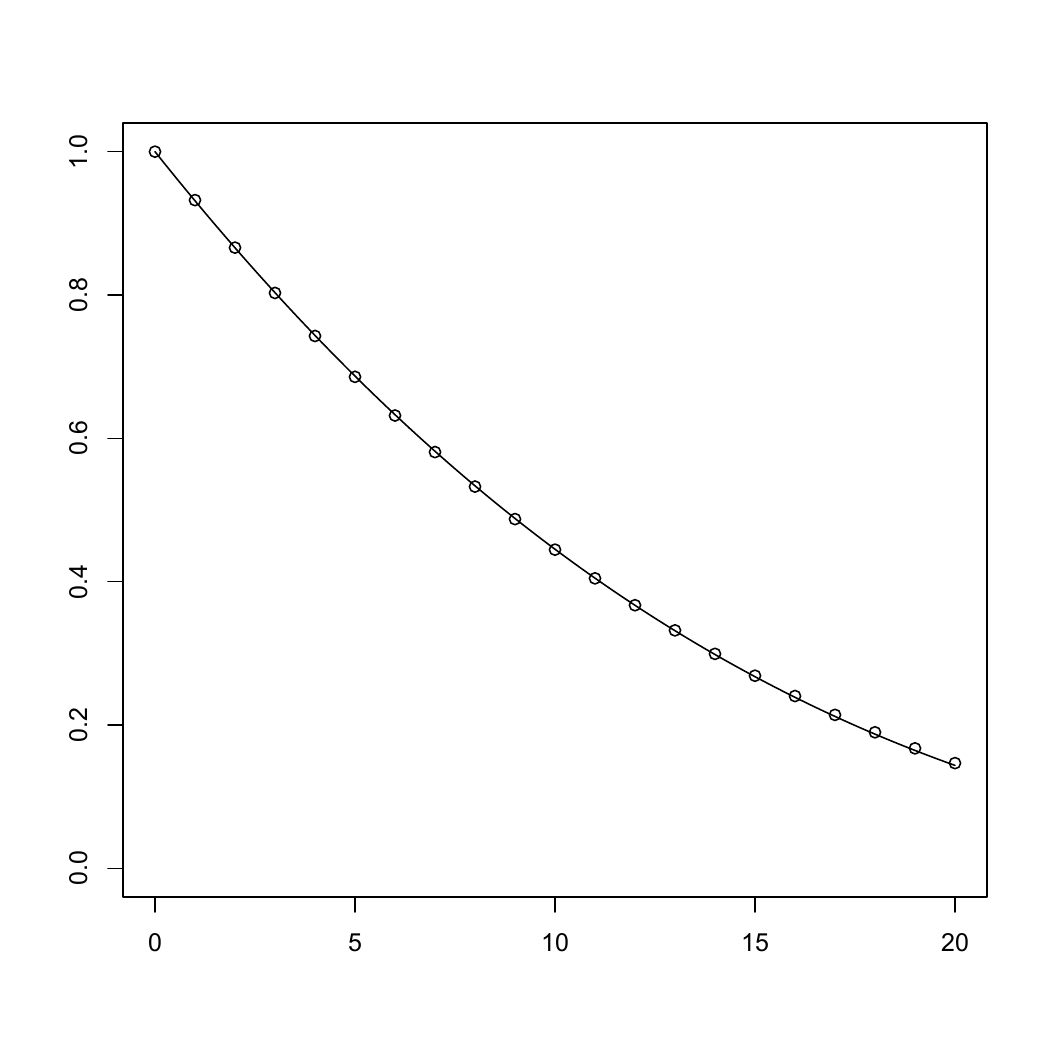} 
\caption{Plot of $k_t$, in a {\bf discrete} (dots) versus {\bf continuous} (curve) model, in a $T=20$ year Riccati tontine, under $x=65$, $m=90$, $b=10$, with $\eta=2\%$, and $\mu=7\%$. The pool size is $n=20$. The frequency at which investors are allowed to leave (and die) doesn't make much of a difference.}
\label{FIG2}
\end{center}
\end{figure}

\end{document}